\documentclass[journal]{IEEEtran}
\usepackage{csquotes}
\usepackage{amsmath}
\usepackage{amsmath, amsthm, amssymb}
\usepackage{algorithm}
\usepackage{algorithmic}
\usepackage{csquotes}
\usepackage{amsmath}

\DeclareMathOperator*{\argmin}{arg\,min}
\usepackage{authblk}
\usepackage{hyperref}
\usepackage{enumitem}
\usepackage{float}
\usepackage{cite}
\usepackage[table,xcdraw]{xcolor}
\usepackage{multirow}

\usepackage[bottom]{footmisc}
\usepackage{mathtools}
\newtheorem{lemma}{Lemma}
\ifCLASSINFOpdf
\usepackage[pdftex]{graphicx}
\else
\fi
\usepackage{amsmath}
\usepackage{amsfonts}
\usepackage{amssymb}
\usepackage{color}
\usepackage{algorithmic}
\usepackage{algorithm}
\usepackage{array}
\usepackage{fixltx2e}
\usepackage{stfloats}
\usepackage{url}
\usepackage{amsmath, amsthm, amssymb}
\newtheorem{corollary}{Corollary}
\newtheorem{proposition}{Proposition}

\usepackage{booktabs}
\usepackage{varwidth}
	
\begin{document}
		\title{STAR-RIS-NOMA Networks: An Error Performance Perspective}

\author{Mahmoud~Aldababsa,~\IEEEmembership{}
	Aymen~Khaleel,~\IEEEmembership{Graduate Student Member,~IEEE}
	and~Ertugrul~Basar,~\IEEEmembership{Senior~Member,~IEEE}
\thanks{\hspace{-0.35cm}
	M. Aldababsa is with the Department of Electrical and Electronics Engineering, Istanbul Gelisim University, 34310, Istanbul, Turkey. (e-mail: mhkaldababsa@gelisim.edu.tr).
	
	A. Khaleel and E. Basar are with the Communications Research and Innovation Laboratory (CoreLab), Department of Electrical and Electronics Engineering, Koc University, Sariyer 34450, Istanbul, Turkey. (e-mail: akhaleel18@ku.edu.tr, ebasar@ku.edu.tr).}
	}

	\maketitle

\begin{abstract}
	This letter investigates the bit error rate (BER) performance of simultaneous transmitting and reflecting reconfigurable intelligent surfaces (STAR-RISs) in non-orthogonal multiple access (NOMA) networks. In the investigated network, a STAR-RIS serves multiple non-orthogonal users located on either side of the surface by utilizing the mode switching protocol. We derive the closed-form BER expressions in perfect and imperfect successive interference cancellation cases. Furthermore, asymptotic analyses are also conducted to provide further insights into the BER behavior in the high signal-to-noise ratio region. Finally, the accuracy of our theoretical analysis is validated through Monte Carlo simulations. The obtained results reveal that the BER performance of STAR-RIS-NOMA outperforms that of the classical NOMA system, and STAR-RIS might be a promising NOMA 2.0 solution.
\end{abstract}


	\IEEEpeerreviewmaketitle

	\section{Introduction}
	
	\IEEEPARstart{N}{on-orthogonal} multiple access (NOMA) has been recognized as a promising multiple access candidate to satisfy the challenging requirements of sixth-generation wireless networks such as high spectral efficiency, massive connectivity, and low latency \cite{NOMA2020}. Different from conventional orthogonal multiple access (OMA), in which different users are allocated to different resource blocks (time, frequency, code), NOMA can accommodate multiple users via the same resource block, which effectively enhances the spectral efficiency. In power-domain (PD)-NOMA, users are assigned to different power levels, and the signals of all users are superimposed into a single message transmitted from the base station (BS). At the receiver side, successive interference cancellation (SIC) is applied to eliminate the inter-user interference and to recover the transmitted symbol\cite{noma2017}. In \cite{BERNOMA1}, the authors analyzed the bit error rate (BER) performance of uplink and downlink NOMA networks, where a binary phase shift keying (BPSK) and quadrature phase shift keying modulation schemes are employed for the far and near users, respectively. In \cite{BERNOMA2}, the closed-form expressions are derived for the BER of a two-user non-orthogonal NOMA system using quadrature amplitude modulation. In \cite{Multi}, the authors proposed a multi carrier-based technique that combines both transmit diversity and NOMA protocol, resulting in enhancing both reliability and sum-rate performance.

Reconfigurable intelligent surfaces (RISs) have also attracted attention as a promising candidate for next generation wireless communication networks. An RIS consists of a massive number of low-cost passive elements that reconfigure the propagation of incident wireless signals by adjusting each element's amplitude and phase shift. In \cite{IoT}, the authors maximized the sum throughput for an RIS-assisted internet-of-things network by jointly optimizing the RIS passive beamforming vector, the transfer time scheduling, and power splitting ratio, under energy harvest-then-transmit policy protocol. In \cite{ber1-ris}-\hspace{-0.01cm}\cite{RIS1}, the authors analyzed the BER for a RIS-assisted NOMA system, where the RIS is partitioned into multiple subsurfaces, each allocated to serve a single user. It is worthing to mention that RISs can only serve the users located on the same side of the surface. Nevertheless, simultaneous transmitting and reflecting RISs (STAR-RISs) are recently proposed to achieve $360^{\circ}$ wireless coverage \cite{STARRIS2021}. Compared to classical RISs, STAR-RISs can serve users located on both sides of their surface. STAR-RISs consist of elements that can produce both electric polarization and magnetization currents, which allows simultaneous control of the transmit and reflect incident signals towards the users located at different sides of the surface. Recently, a number of works have been reported incorporating NOMA and STAR-RIS in wireless networks \cite{star32021}-\hspace{-0.01cm}\cite{star182021}. The authors in \cite{star32021} showed that the coverage can be significantly extended by integrating STAR-RISs into NOMA networks. The authors in \cite{star52021} took advantage of STAR-RIS to simultaneously eliminate  inter-cell interference and enhance the desired signals in NOMA enhanced coordinated multi-point transmission networks. In \cite{weighted-sum}, \citen{secrecy} and \citen{star62021}, the authors maximized the weighted sum-rate for the STAR-RIS-assisted multiple-input multiple-output system, the sum secrecy rate for the STAR-RIS-assisted multiple-input single-output system, and the achievable sum-rate of STAR-RIS-NOMA system, respectively. The STAR-RIS in \cite{star82021} is utilized for adjusting users' decoding order to efficiently mitigate the mutual interference between users and to extend the coverage of heterogeneous networks. In \cite{star92021} and \cite{star102021}, the authors proposed different approximated mathematical channel models to investigate the outage probability (OP) performance for STAR-intelligent omini-surfaces based NOMA and STAR-RIS-NOMA multi-cell networks. In \cite{star122021}, the authors considered a realistic transmission and reflection coupled phase-shift model for STAR-RISs. The authors in \cite{star2021} proposed a STAR-RIS partitioning algorithm for a STAR-RIS-NOMA network. It is designed to maximize the sum-rate and guarantee the quality-of-service requirements by assigning a proper number of STAR-RIS elements to each user. In \cite{star132021} and \cite{star142021}, optimizations problems are considered to minimize the power consumption and maximum secrecy OP, respectively, for uplink STAR-RIS-NOMA networks. The authors in \cite{star152021} investigated the resource allocation scheme in STAR-RIS-aided multi-carrier OMA and NOMA networks to maximize the sum-rate. In \cite{star162021}, a deep reinforcement learning-based algorithm is designed to maximize the energy efficiency for a STAR-RIS-NOMA network. The authors in \cite{star172021} proposed practical phase-shift configuration strategies for STAR-RIS-NOMA networks with correlated phase shifts. In \cite{star182021}, the performance of STAR-RIS-NOMA networks is investigated in terms of OP and ergodic sum-rate.

To the best of the authors' knowledge, the studies of \cite{BERNOMA1}-\hspace{-0.01cm}\cite{BERNOMA2} assumed a simple scenario of two NOMA users, which might not be practical. Furthermore, the works of \cite{ber1-ris}-\hspace{-0.01cm}\cite{RIS1} assumed that each user receives the signals reflected from the RIS partition assigned to it only and without interference signals from the other partition, which is physically impossible. Moreover, the performance of STAR-RIS-NOMA networks in \cite{star32021}-\hspace{-0.01cm}\cite{star182021} is analyzed in terms of OP and sum-rate only, and no light has been shed on the BER  performance aspect yet. These have motivated us to investigate the BER performance for STAR-RIS-NOMA networks, where BS communicates multiple non-orthogonal users with the assistance of a STAR-RIS. The main contributions of the paper can be summarized as follows. We analyze the BER performance of the STAR-RIS-NOMA network considering the subsurface mutual interference within the RIS transmission or reflection part in the cases of perfect and imperfect SIC. We derive the closed-form expressions of the BER for the users employing a BPSK modulation. Then, asymptotic analyses are carried out to provide further insight into BER behavior in the high signal-to-noise ratio (SNR) region. Finally, we verify our analytical results by Monte Carlo simulations, which demonstrate the superiority of the proposed network over the classical NOMA system.\footnote{{\it Notation:} Matrices and vectors are denoted by an upper and lower case boldface letters, respectively. $\mathbb{C}^{m\times n}$ denotes the set of matrices with dimension $m\times n$. $|\cdot|$, $\mathrm{E}[\cdot]$, and $(\cdot)^{T}$ represent the absolute value, expectation operator, and transpose operator, respectively. $f_{X}\left(x\right)$ denotes probability density function (PDF) of a random variable (RV) $ X $. $ \mathcal{CN}(\mu,\sigma^{2})$ stands for the complex Gaussian distribution with mean $\mu$ and variance $\sigma^{2}$. $Q\left(\cdot\right)$ and $ erf (\cdot) $ denote the $ Q $-function and error function, respectively.}
	

	\section{System Model}
	\label{sec:systemmodel}

	\begin{figure}[]
	\includegraphics[width=.33\textwidth]{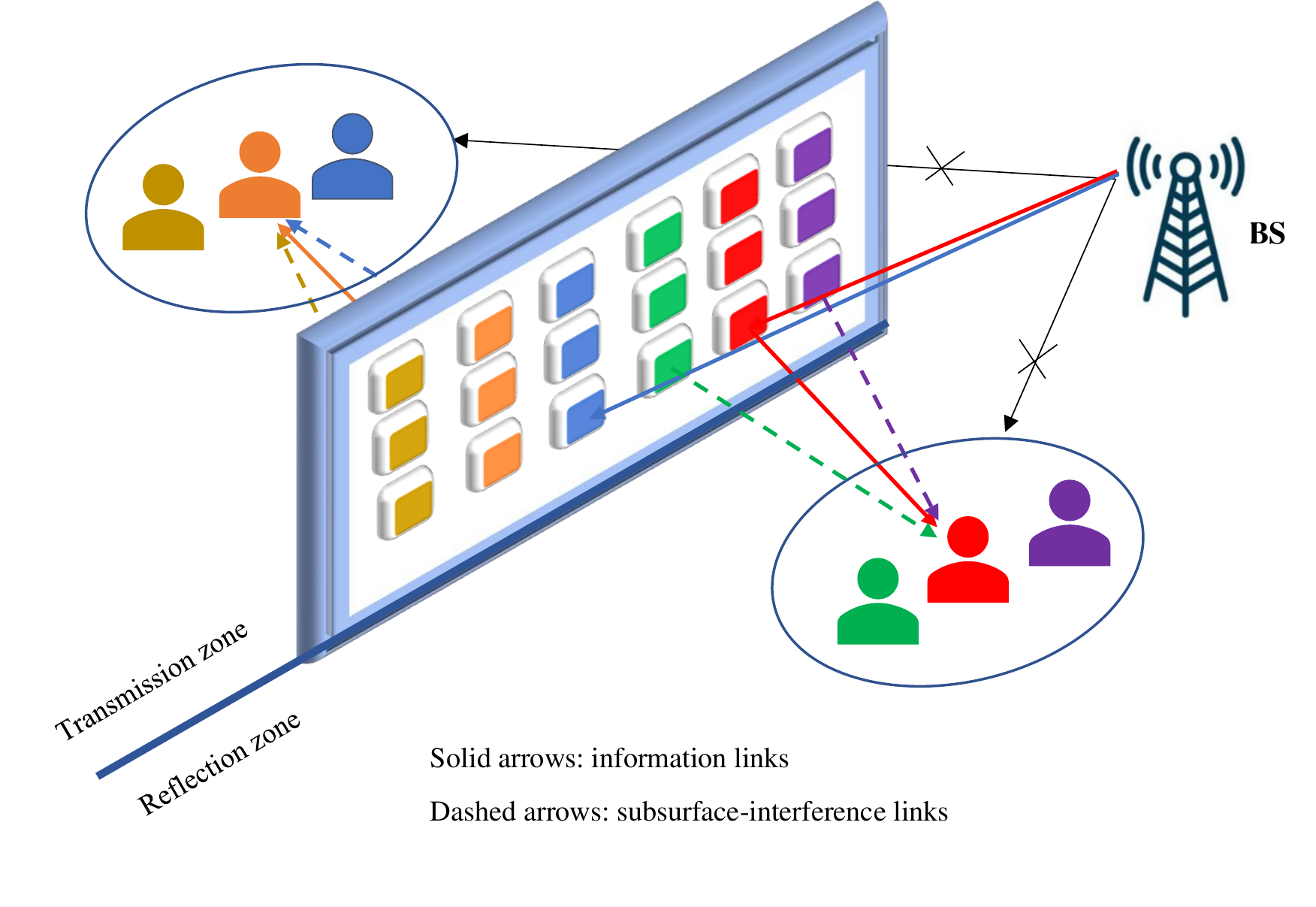}
	\centering
	\caption{A STAR-RIS-assisted NOMA system model.}	
	\label{System model}
	\vspace{-0.8cm}
\end{figure}

	As illustrated in Fig. \ref{System model}, we consider a STAR-RIS-NOMA network, in which a single-antenna BS communicates simultaneously $K$ single-antenna users with the help of a STAR-RIS equipped with $N$ passive elements. The $K$ users $\left(\text{U}_{1}, ..., \text{U}_{K}\right)$ are deployed on both sides of the STAR-RIS, where $K_{t}$ and $K_{r}$ users are located in the transmission and reflection zones, respectively. The direct communication links between the BS and users are assumed to be unavailable due to natural or manmade obstacles and the STAR-RIS is deployed to provide alternative communication links for users through the transmission/reflection elements. The STAR-RIS adopts mode switching (MS) protocol, where each element can operate in full transmission or reflection mode. Thus, the STAR-RIS is partitioned into two main parts, where the first and second parts contain $N_t$ transmitting and $N_r$ reflecting elements to serve the users in transmission and reflection zones, respectively. Both parts of the RIS are partitioned into subsurfaces, where each subsurface is allocated to serve a specific user in the transmission or reflection zone. The BS-STAR-RIS link is assumed to follow Rayleigh fading channel model, where $\mathbf{h}_{i}\in \mathbb{C}^{N_{i}\times 1}$ denotes the BS-$i$th subsurface vector with $ N_{i} $ refers to the number of elements that belong to the $i$th subsurface in the transmission ($t$) or reflection ($r$) part of the STAR-RIS. $\mathbf{h}_{i} = \left[h^{(1)}_{i}, ..., h^{(n)}_{i}, ..., h^{(N_{i})}_{i}\right]^{T}$ with $h^{(n)}_{i}=\sqrt{L_{BS}}\zeta^{(n)}_{i}e^{-j\phi^{(n)}_{i}}$ denoting the channel coefficient between the BS and $n$th STAR-RIS element in the $i$th subsurface, where $L_{BS}$, $\zeta^{(n)}_{i}$, and $\phi^{(n)}_{i}$ denote the path gain, channel amplitude, and channel phase, respectively, and $h_i^{(n)}\sim \mathcal{CN}(0,L_{BS})$. Likewise, the STAR-RIS-$\text{U}_{k}$ link is assumed to follow Rayleigh fading channel model, where the channel vector between the $i$th subsurface in the transmission or reflection part of the STAR-RIS and $\text{U}_{k}$ is denoted by $\mathbf{g}_{k,i}\in \mathbb{C}^{N_{i}\times1}$, $\mathbf{g}_{k,i} = \left[g^{(1)}_{k,i}, ..., g^{(n)}_{k,i}, ..., g^{(N_{i})}_{k,i}\right]^{T}$, $g^{(n)}_{k,i} = \sqrt{L_{SU,k}}\eta^{(n)}_{k,i}e^{-j\Phi^{(n)}_{k,i}}$ denotes the channel coefficient between $n$th RIS element in the $i$th subsurface and $\text{U}_{k}$, where $L_{SU,k}$, $\eta^{(n)}_{k,i}$, and $\Phi^{(n)}_{k,i}$ denote the path gain, channel amplitude, and channel phase, respectively, and $g^{(n)}_{k,i}\sim \mathcal{CN}(0,L_{SU,k})$. The transmission and reflection coefficients for the $i$th subsurface in the transmission and reflection parts of the STAR-RIS are denoted by the entries of the diagonal matrix $\mathbf{\Theta}_{i}\in \mathbb{C}^{N_{i}\times N_{i}}$, for the $n$th element we have $\Theta^{(n,n)}_{i}=e^{j\theta^{(n)}_{i}}$, where $\theta^{(n)}_{i}\in[0,2\pi)$. In the considered setup, the users are ordered according to their channel gains such $\text{U}_1$ and $\text{U}_K$ are the weakest and strongest channel gains, respectively. Thus, the power allocation are inversely ordered as $a_{1}\ge...\ge a_{K}$. Given, $P$ is the BS transmit power, $x_k$ and $a_k$ are $\text{U}_k$'s symbol and power allocation factor, respectively, ${\mathrm{E}}\left[|x_{k}|^{2}\right]=1$, $\sum_{k}a_{k} = 1$. Then, the BS transmits the superimposed signal $x = \sum_{k}\sqrt{a_{k}P}x_{k}$ to all users and the received signal at $\text{U}_k$ is 
	\begin{align} 
		\label{receivedsignal}
		y_{k} &= \mathbf{g}_{k}^{T}\mathbf{\Theta}_{k}\mathbf{h}_{k}x + \sum_{i}^{}\mathbf{g}_{k,i}^{T}\mathbf{\Theta}_{i}\mathbf{h}_{i}x +  n_{k},
	\end{align}
	where $\sum_{i}^{}\mathbf{g}_{k,i}^{T}\mathbf{\Theta}_{i}\mathbf{h}_{i}x$ refers to subsurfaces mutual interference within the transmission or reflection part of the RIS and ${n}_{k}$ denotes the complex additive white Gaussian noise (AWGN) sample with zero mean and variance $\sigma^{2}$ at $\text{U}_k$, i.e., ${n}_{k}\sim \mathcal{CN}(0,\sigma^{2})$. Notably, $\text{U}_1$ has the highest allocated power coefficient, and thus it does not perform the SIC process. It detects its signal directly by considering other users' symbols as noise and implement maximum likelihood detection (MLD). On the other hand, $\text{U}_k$ needs the SIC process to detect and subtract $\left\{1,  ..., k-1\right\}$'s signals and then it implements MLD to detect its signal. Accordingly, the detected $\hat{x}_{k}$ can be stated as
	\begin{align}
	\label{Userk_ML}
	\hat{x}_{k} =\underset{i}{{\argmin}}\left\{\left|y^{*}_{k} - \sqrt{a_{k}P}\mathbf{g}_{k}^{T}\mathbf{\Theta}_{k}\mathbf{h}_{k}x^{(i)}_{k}\right|^{2}\right\},
\end{align}
where $y^{*}_{k} = y_{k} - \sum_{j=1}^{k-1}\sqrt{a_{j}P}\mathbf{g}_{k}^{T}\mathbf{\Theta}_{k}\mathbf{h}_{k}\hat{x}_{j}$ and $x^{(i)}_{k}$ denotes $i$th element in the set of $x_{k}$ constellation.
\section{Bit Error Rate Analysis}	
\label{sec:beranalysis}

The average bit error probability for $\text{U}_k$ is given by
\begin{align}
	\label{Pe_integralk}
	\mathcal{P}_{e,k} =\int_{0}^{\infty}\mathcal{P}(e|\varphi_{k})f_{\varphi_{k}}(x)dx,
\end{align}
where $\varphi_{k} = |\mathbf{g}_{k}^{T}\mathbf{\Theta}_{k}\mathbf{h}_{k}|$, $\mathcal{P}(e|\varphi_{k}) $ and $ f_{\varphi_{k}}(x) $ denote absolute cascaded channel, conditional BER for $\text{U}_k$ in AWGN channel and the PDF of the cascaded channel $\varphi_{k}$, respectively. In what follows, we derive $\mathcal{P}(e|\varphi_{k}) $ and $ f_{\varphi_{k}}(x) $ in Lemmas 1 and 2, respectively, then we obtain $\mathcal{P}_{e,k}$ in the cases of absence and presence of SIC error in Propositions 1 and 2, respectively.	 
\begin{lemma}
	The BER for $\text{U}_k$ in AWGN can be given by
	\begin{align} 
		\label{AWGNuser_k}
		\mathcal{P}\left(e|\varphi_{k}\right) & = \sum_{i}^{}\mathcal{P}\left(A^{(i)}_{k} \right)Q\left(A^{(i)}_{k} \varphi_{k}\sqrt{\varrho_{k}\gamma}\right),
	\end{align}	 
	where $A^{(i)}_{k} $ denotes $i$th combination in $A_{k} \in \left\{\sqrt{a_{k}P} \pm ... \pm \sqrt{a_{K}P}\right\}$ and $\mathcal{P}\left(A^{(i)}_{k}\right) = 2^{k-K}$. Additionally, $\gamma \overset{\bigtriangleup}{=} \frac{P}{\sigma^{2}}$ is the signal-to-noise ratio (SNR), $ \varrho_{k} = \left({1 + L_{k}\left(N_{\chi}-N_{k}\right)\gamma/P}\right)^{-1} $. Here, $ \chi\in\left\{t,r\right\} $, $L_{k}={L_{BS}L_{SU,k}}={d_{BS}^{-\alpha_{BS}}d_{SU,k}^{-\alpha_{SU}}}$ is the overall path gain of the BS-RIS-$\text{U}_k$ link, where $\left(\alpha_{BS}, d_{BS}\right) $, and $\left(\alpha_{SU}, d_{SU,k}\right)$ refer to $\left(\text{path loss exponents, distances}\right)$ associated with the BS-RIS and RIS-$U_k$ links, respectively.
\end{lemma}

\begin{proof} 
According to the central limit theorem (CLT), for large number of RIS elements, $\sum_{i}^{}\mathbf{g}_{k,i}^{T}\mathbf{\Theta}_{i}\mathbf{h}_{i}x$ in \eqref{receivedsignal} converges to Gaussian RV, i.e., $ \sum_{i}^{}\mathbf{g}_{k,i}^{T}\mathbf{\Theta}_{i}\mathbf{h}_{i}x\sim \mathcal{CN}\left(0,L_{k}\left(N_{\chi}-N_{k}\right)\right)$. Then, $y_{k} = \mathbf{g}_{k}^{T}\mathbf{\Theta}_{k}\mathbf{h}_{k}x + w_{k}$, where $w_{k}\sim \mathcal{CN}(0,\sigma^{2} + L_{k}\left(N_{\chi}-N_{k}\right))$. Following same derivation steps in \cite{BERNOMA1}-\hspace{-0.01cm}\cite{BERNOMA2}, the BER for $\text{U}_k$ in AWGN can be achieved in \eqref{AWGNuser_k}.
\end{proof}

\begin{lemma}
	The PDF of $\varphi_{k}$ is given by
	\begin{align} 
		\label{phi_k}
		f_{\varphi_{k}}(x) =\frac{1}{\sqrt{2\pi v_{k}}}\exp\left(-\frac{\left(x-\mu_{k}\right)^{2}}{2v_{k}}\right),
	\end{align}	
where $\mu_{k} = \frac{\pi}{4}\sqrt{L_{k}}N_{k}$ and $v_{k} = \left(1-\frac{\pi^{2}}{16}\right)L_{k}N_{k}$.
\end{lemma}
\begin{proof} 

	 The SNR can be maximized by letting $ \theta^{(n)}_{k} = \phi^{(n)}_{k} + \Phi^{(n)}_{k} $. Thus, $\varphi_{k}$ corresponds to the sum of product of two independent Rayleigh RVs. According to the CLT, when $N_{k}\rightarrow \infty$, $\varphi_{k}$ follows Gaussian distribution, $\varphi_{k}\sim \mathcal{N}\left(\mu_{k},v_{k}\right)$. Hence, the PDF of $\varphi_{k}$ can be finally expressed as in \eqref{phi_k}.
\end{proof}

\begin{proposition} In the perfect SIC case, the BER of $\text{U}_k$ is 
\begin{align}
	\label{BER_userk}
	\mathcal{P}_{e,k} & \approx \frac{\exp\left(-\frac{2cv_{k}+\mu^{2}_{k}}{2v_{k}}\right)}{\sqrt{2 v_{k}}}\sum_{i}\mathcal{P}\left(A^{(i)}_{k}\right)\nonumber\\
	&\exp\left(\frac{\left(bA^{(i)}_{k}\sqrt{\varrho_{k}\gamma}v_{k}-\mu_{k}\right)^{2}}{4a\left(A^{(i)}_{k}\right)^{2}\varrho_{k}\gamma v^{2}_{k}+2v_{k}}\right)\sqrt{\frac{v_{k}}{4\left(A^{(i)}_{k}\right)^{2}\varrho_{k}\gamma v_{k} + 2}}\nonumber\\
	&\left(1 + \text{erf}\left(\sqrt{\frac{\left(bA^{(i)}_{k}\sqrt{\varrho_{k}\gamma}v_{k}-\mu_{k}\right)^{2}}{4a\left(A^{(i)}_{k}\right)^{2}\varrho_{k}\gamma v^{2}_{k}+2v_{k}}}\right)\right),
\end{align}
where $ \left(a,b,c\right) = \left( 0.3842, 0.7640, 0.6964\right) $ are the fitting coefficients used by Lopez-Benitez and Casadevall in approximating $ Q $-function \cite{Qfunction}.
\end{proposition}
\begin{proof}
 A tighter and more tractable approximation of $ Q $-function in Lemma 1 can be given by Lopez-Benitez and Casadevall. The integral obtained by substituting this approximation and Lemma 2 into \eqref{Pe_integralk} can be written in terms of error function \cite[Eq. (3.322.2)]{Grad}. Therefore, the BER of $\text{U}_k$ can be finally stated in \eqref{BER_userk}. 
\end{proof}

In general scenarios, detection errors can appear at any step of the SIC process. Proposition 2 considers a special case (two-user scenario, where $\text{U}_1$ and $\text{U}_2$ are located at the transmission and reflection zones, respectively), which can be generalized for $K$ users. 

\begin{proposition} In the presence of SIC error, the BER for $\text{U}_1$ and $\text{U}_2$ can be given by
	\begin{align}
		\label{SIC}
		&\mathcal{P}^{sic}_{e,1} = \mathcal{P}_{e,1}\hspace{0.08cm} \text{and}  \hspace{0.08cm}   \mathcal{P}^{sic}_{e,2} \approx \mathcal{P}_{e,2}\mathcal{P}\left(x^{c}_{1}\right) + 0.5\left(1-\mathcal{P}\left(x^{c}_{1}\right)\right),
	\end{align}
where $\mathcal{P}\left(x^{c}_{1}\right)$ refers to average bit error probability for $\text{U}_1$ when its signal is detected at $\text{U}_2$ and can be calculated from \eqref{BER_userk} with related parameters of $\varphi_{2}$.
\end{proposition}
\begin{proof} 
		$\text{U}_1$ does not carry out SIC process and hence the BER for $\text{U}_1$ is calculated form \eqref{BER_userk}, i.e., $\mathcal{P}^{sic}_{e,1} = \mathcal{P}_{e,1}$. On the other hand, $\text{U}_2$ performs SIC process, and hence the BER for $\text{U}_2$ should be considered for both correct and erroneous SIC cases, $ \mathcal{P}^{sic}_{e,2} = \mathcal{P}\left(x_{2}/x^{c}_{1}\right)\mathcal{P}\left(x^{c}_{1}\right) + \mathcal{P}\left(x_{2}/x^{e}_{1}\right)\mathcal{P}\left(x^{e}_{1}\right)$. Here, $ \mathcal{P}\left(x^{c}_{1}\right) $ and $ \mathcal{P}\left(x^{e}_{1}\right) $ denote average bit error probability for correct and erroneous detection of $x_{1}$ at $\text{U}_2$, respectively. $ \mathcal{P}\left(x_{2}/x^{c}_{1}\right) $ and $ \mathcal{P}\left(x_{2}/x^{e}_{1}\right) $ denote average bit error probability for $x_{2}$ when $x_{2}$ is correctly and erroneously detected at $\text{U}_2$, respectively. $\mathcal{P}\left(x^{c}_{1}\right)$ is calculated from \eqref{BER_userk} with related parameters of $\varphi_{2}$, and hence $\mathcal{P}\left(x^{e}_{1}\right) = 1- \mathcal{P}\left(x^{c}_{1}\right)$. Form \eqref{BER_userk}, $ \mathcal{P}\left(x_{2}/x^{c}_{1}\right) $ can be also found, i.e., $\mathcal{P}\left(x_{2}/x^{c}_{1}\right) = \mathcal{P}_{e,2}$ and $ \mathcal{P}\left(x_{2}/x^{e}_{1}\right) $ is very close to the worst case, i.e., $ \mathcal{P}\left(x_{2}/x^{e}_{1}\right) \approx 0.5$.
\end{proof}

It is difficult to get an insight from the BER expressions in Propositions. In corollary 1, we examine the BER behavior in the high SNR regime.

\begin{corollary} In high SNR region, the BER can be given as
	\begin{align}
		\label{Errorfloor}
		\mathcal{P}^{\infty}_{e,k} & \approx \frac{\exp\left(-\frac{2cv_{k}+\mu^{2}_{k}}{2v_{k}}\right)}{\sqrt{2 v_{k}}}\sum_{i}\mathcal{P}\left(A^{(i)}_{k}\right)\nonumber\\
		&\times\exp\left(\frac{\left(bA^{(i)}_{k}\sqrt{\varpi}v_{k}-\mu_{k}\right)^{2}}{4a\left(A^{(i)}_{k}\right)^{2}\varpi v^{2}_{k}+2v_{k}}\right)\sqrt{\frac{v_{k}}{4\left(A^{(i)}_{k}\right)^{2}\varpi v_{k} + 2}}\nonumber\\
		&\times\left(1 + \text{erf}\left(\sqrt{\frac{\left(bA^{(i)}_{k}\sqrt{\varpi}v_{k}-\mu_{k}\right)^{2}}{4a\left(A^{(i)}_{k}\right)^{2}\varpi v^{2}_{k}+2v_{k}}}\right)\right),
	\end{align}
	where $\varpi = \left(L_{k}\left(N_{\chi}-N_{k}\right)/P\right)^{-1} $.
\end{corollary}

\begin{proof}
	Let $\varpi = \varrho_{k}\gamma$, when $\gamma\rightarrow \infty$, $\varpi =  \left(L_{k}\left(N_{\chi}-N_{k}\right)/P\right)^{-1}$, substituting in \eqref{BER_userk} we obtain \eqref{Errorfloor}.
\end{proof}

From Corollary 1, it can be noticed that, due to the subsurface interference, the asymptotic BER does not depend on the SNR. The BER reaches a fixed value (error floor) in the high SNR region, which indicates that the diversity gain is zero.



	\section{Numerical Results}
	\label{sec:results}
		
	In this section, Monte Carlo simulations are presented to show the performance of the STAR-RIS-NOMA system, and to validate the BER analytical results. Unless stated otherwise, we assume $ P = 1 $ and $\alpha_{BS} = \alpha_{SU} = 2$. We compare the proposed system with the classical NOMA system by considering the same simulation parameters, and the BS-$\text{U}_k$ path gain is modelled as $L_k=d_k^{-\alpha}$, where $d_k$ is the BS-$\text{U}_k$ distance and $\alpha=2$ is the path loss exponent. Fig. \ref{BERvsSNR}(a)-(b) show the BER performance versus SNR in the case of perfect SIC for two-user scenario, where each user is located at different side of RIS. It can be clearly seen that the obtained analytical results perfectly match the simulations curves. It can be also seen that both users achieve better BER performance as the number of RIS elements increases. Compared to classical NOMA system, Fig. \ref{BERvsSNR}(b) shows that  STAR-RIS NOMA ($\text{U}_2$) achieves $20$ dB, $25$ dB, and $29$ dB performance gains for $N_{1}=25, 50$ and $75$, respectively at BER of $10^{-3}$. Fig. \ref{BERvsSNRsic} shows the BER performance versus SNR for $\text{U}_{2} $ in the cases of perfect and imperfect SIC. It is observed that the BER performance gets worse in the presence of SIC error even when the number of RIS elements increases. However, when $a_2$ decreases ($a_1$ increases), the SIC error decreases. This is because increasing $a_1$ increases the probability of detecting its signal correctly and results in a reduction of SIC error. Fig. \ref{BERvselements} depicts the BER performance versus number of RIS elements for different power allocation coefficients. It is observed that $\text{U}_1$ is more sensitive to the change in the power allocation over the change of the number of RIS elements. This is because $\text{U}_1$ does not carry out any SIC process, and hence changes in the power allocation result in variations in the interference level caused by $\text{U}_2$. For instance, an increment of $0.1$ in the power allocation makes $\text{U}_1$ need $13$ fewer elements to achieve the same BER of $10^{-3}$. On the other side, $\text{U}_2$ needs 3 elements more to achieve the same BER of $10^{-3}$. Fig. \ref{BERthreeusers} plots the BER performance versus SNR for three-user scenario, where $\text{U}_1$ and $\text{U}_2$ are located at the transmission zone and $\text{U}_3$ is located at the reflection zone. Notice that BER performance enhances as the number of RIS elements increases. However, when SNR increases, the BERs for $\text{U}_1$ and $\text{U}_2$ reach an error floor due to the mutual subsurface-interference. On the other hand, the BER for $\text{U}_3$ does not reach error floor since $U_{3}$ does not suffer from any mutual subsurface-interference.

\begin{figure}[]
	\includegraphics[width=0.4\textwidth]{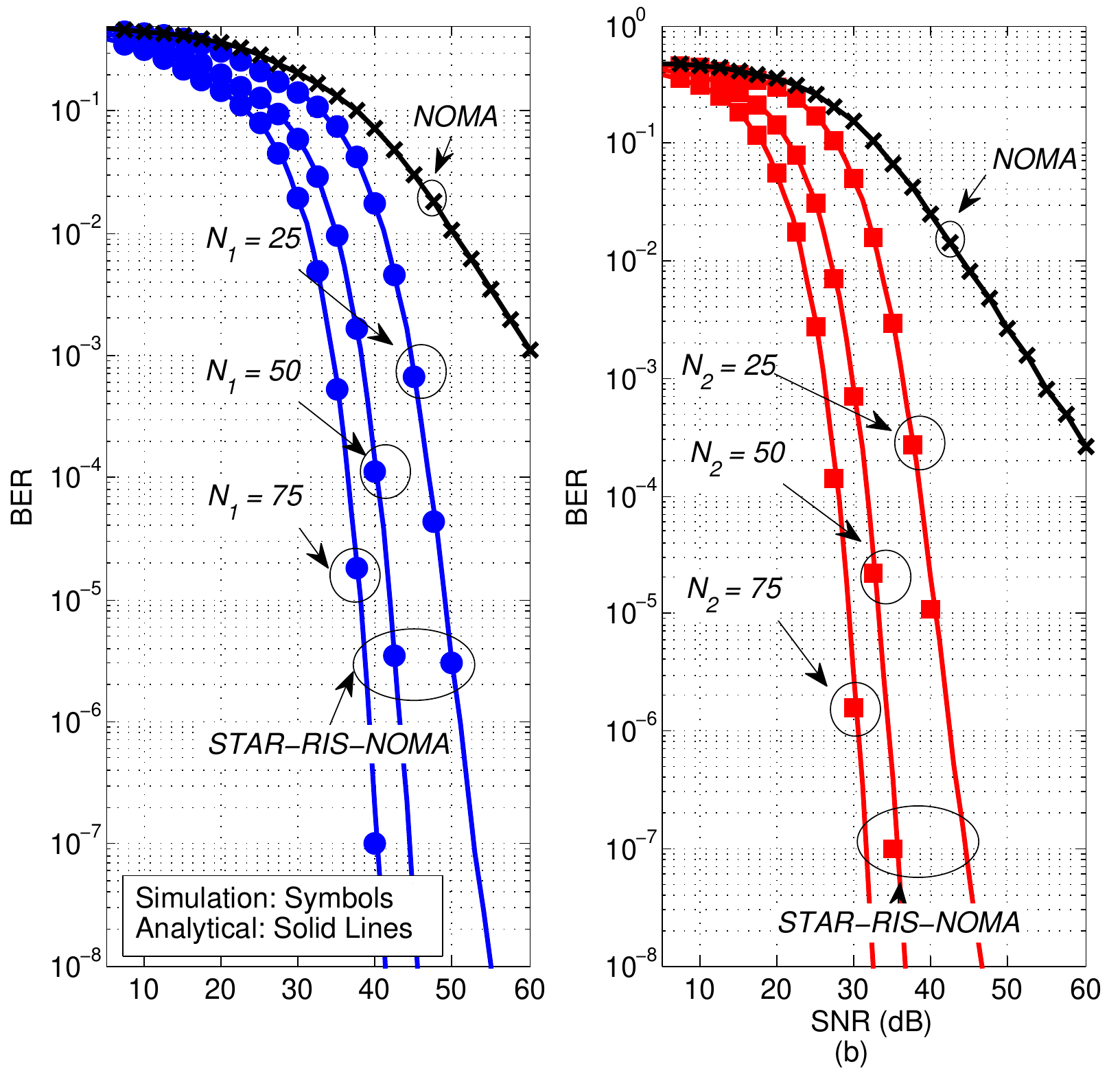}
	\centering
	\caption{The BER versus SNR for (a) $ \text{U}_{1} $ and (b) $ \text{U}_{2} $ in the perfect SIC case, $\left(a_{1}, a_{2}\right) = \left(0.7, 0.3\right)$ and $\left(d_{BS}, d_{SU,1}, d_{SU,2}\right) = \left(50, 6, 4\right)$ m}	
	\label{BERvsSNR}
\end{figure}

\begin{figure}[]
	\includegraphics[width=0.4\textwidth]{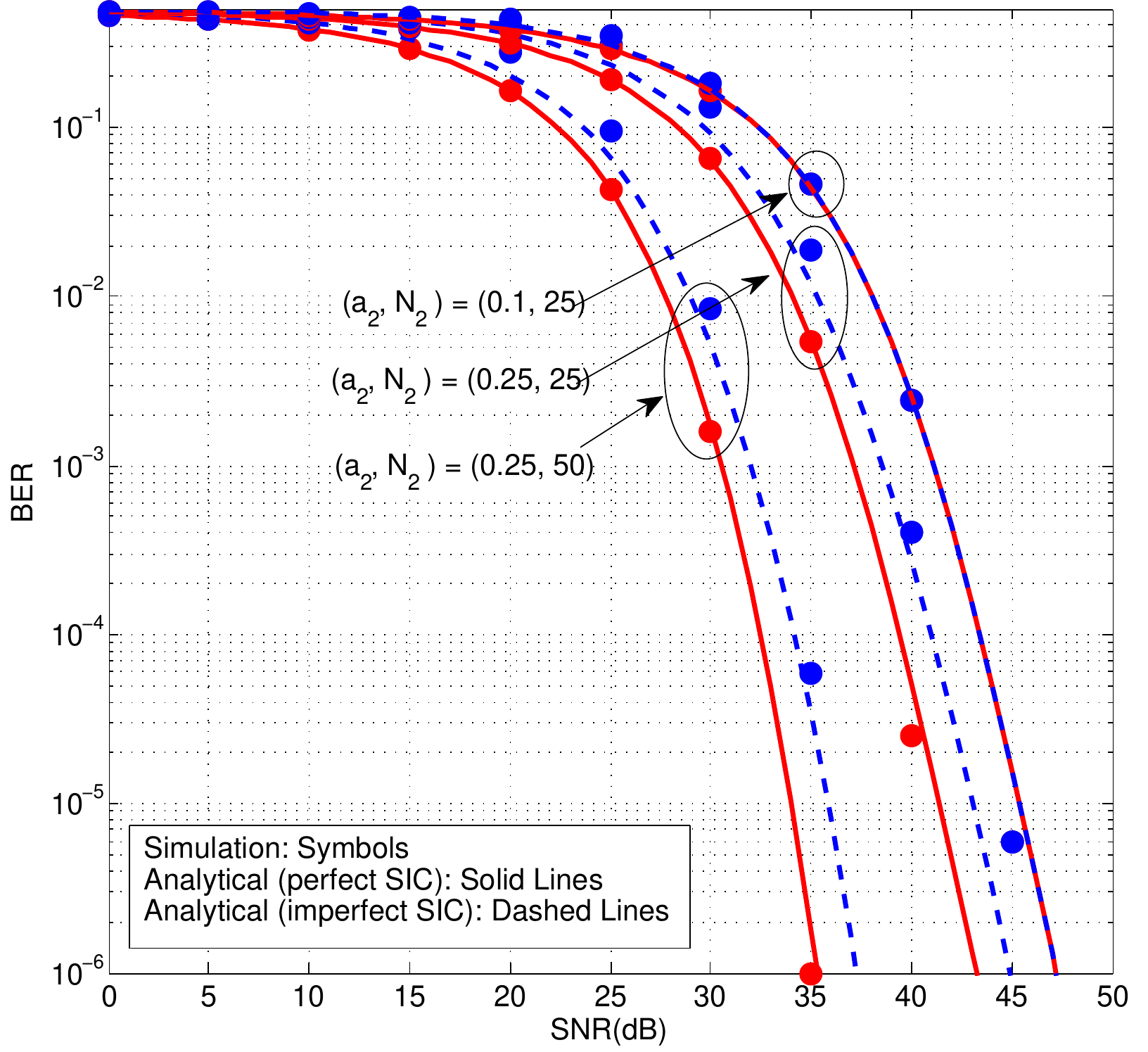}
	\centering
	\caption{The BER performance versus SNR for $\text{U}_{2} $ in the perfect and imperfect SIC cases.}	
	\label{BERvsSNRsic}
\end{figure}
\begin{figure}[]
	\includegraphics[width=0.4\textwidth]{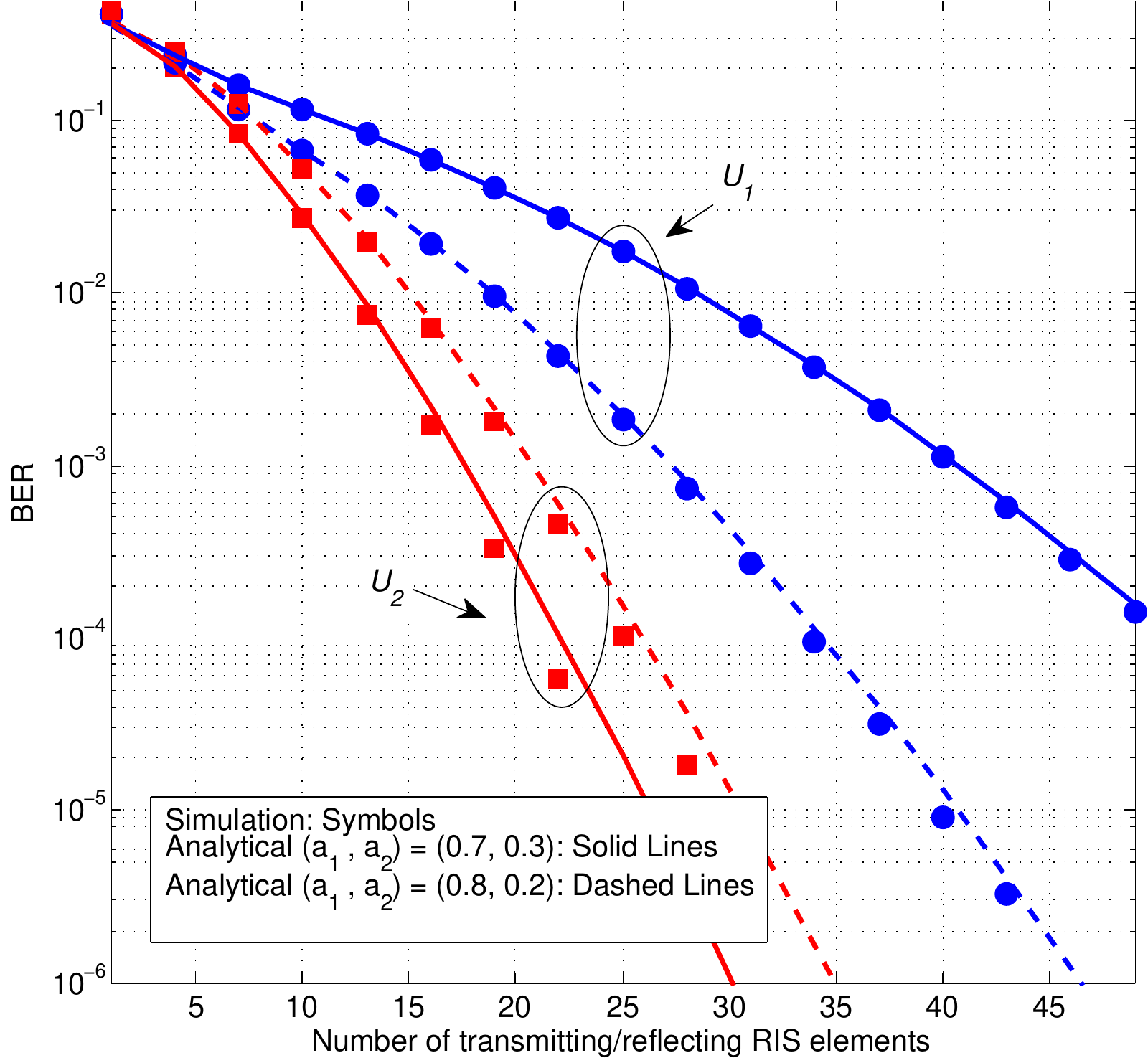}
	\centering
	\caption{The BER performance versus number of transmitting/reflecting elements for different power allocation coefficients and $\gamma = 40$ dB.}	
	\label{BERvselements}
\end{figure}
\begin{figure}[]
	\includegraphics[width=0.4\textwidth]{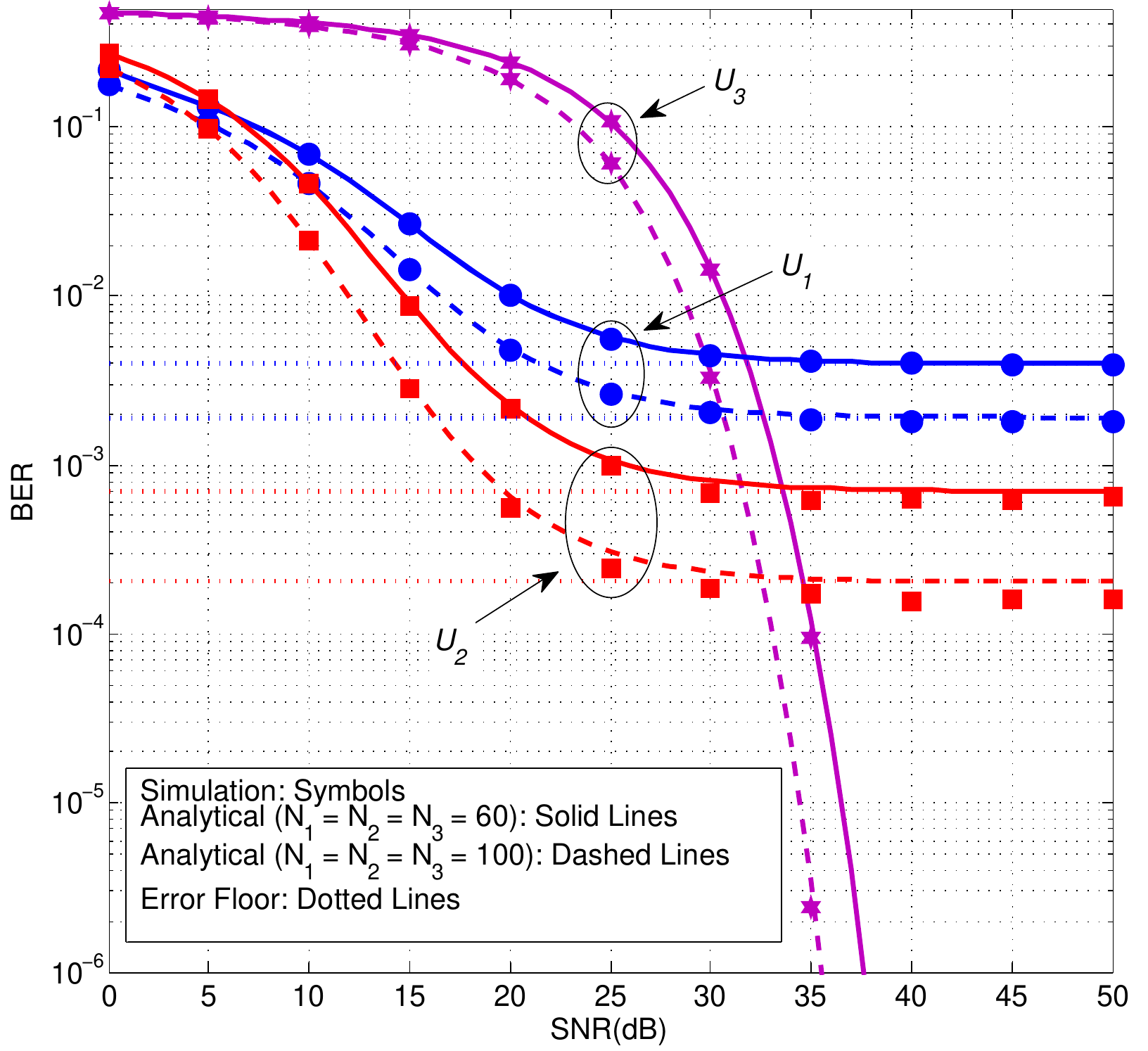}
	\centering
	\caption{The BER versus SNR for three-user scenario, $\left(a_{1}, a_{2}, a_{3}\right)  = \left(0.75, 0.248, 0.002\right)$ and $\left(d_{BS}, d_{SU,1}, d_{SU,2}, d_{SU,2}\right) = \left(20, 3, 2.5, 2\right)$ m.}	
	\label{BERthreeusers}
\end{figure}

\section{Conclusion}
\label{sec:conclusion}
	We have examined the BER performance of STAR-RIS-NOMA network, in which a STAR-RIS utilizes the MS protocol to provide communication between a BS and multiple NOMA users located on both its sides. We have derived closed-form expressions for BER in both perfect and imperfect SIC cases. Then, asymptotic analyses have been carried out to provide further insights into the error rate behavior in the high SNR region. Based on the numerical results, we conclude that the BER performance enhances with increasing numbers of RIS elements. Moreover, in the case of imperfect SIC, the BER performance gets worse. On the other hand, in the case of perfect SIC, the users achieve diversity gain. Furthermore, the farthest user is more sensitive to changes in the power allocation over the change of numbers of RIS elements. Finally, it has been shown that the BER performance of STAR-RIS-NOMA outperforms that of the classical NOMA system and STAR-RIS based NOMA can be useful candidate for future systems. For a future work, the BER analysis provided here can be extended to the scenario where the STAR-RIS is partitioned between multiple users, and the aim is to find the minimum number of elements needs to be allocated for each user to guarantee a minimum BER threshold for all users.

	\end{document}